\documentclass[10pt]{article}
\usepackage{mathrsfs}
\usepackage{amsfonts}
\usepackage{amsthm,amsmath,amssymb,anysize}
\newtheorem{lemma}{Lemma}[section]
\newtheorem{theorem}[lemma]{Theorem}
\newtheorem{remark}[lemma]{Remark}
\newtheorem{proposition}[lemma]{Proposition}

\newtheorem{corollary}[lemma]{Corollary}
\newtheorem{example}[lemma]{Example}

\usepackage[numbers,sort&compress]{natbib}
\marginsize{4.2cm}{4.2cm}{3.6cm}{3cm}%×óÓÒÉÏÏÂ
\numberwithin{equation}{section}

\title{\textsf{Generators of simple Lie superalgebras in characteristic
zero}}
\author{\textsc{ Wende Liu$^{1,2}$}\footnote{Supported by the NSF for Distinguished Young Scholars, HLJ Province (JC201004) and  the NSF
  of China (10871057)}\ \textsc{and}
    \textsc{Liming Tang$^{1,2}$}\footnote{Correspondence:  \texttt{wendeliu@ustc.edu.cn}  (W. Liu), \texttt{limingaaa2@sina.com} (L. Tang)
 }
  \\
  \small\textit{$^{1}$Department of Mathematics},
  \small\textit{Harbin Institute of Technology}\\
  \small\textit{Harbin 150006, China}\\
 \small\textit{$^{2}$School of Mathematical Sciences},
 \small \textit{Harbin Normal University} \\
 \small \textit{Harbin 150025, China}}
\date{ }
\begin{document}
\maketitle
\begin{quotation}
{\small\noindent \textbf{Abstract}:
It is shown that any finite dimensional simple Lie superalgebra over an
algebraically closed field of characteristic 0 is generated by 2
elements.

\vspace{0.05cm} \noindent{\textbf{Keywords}}: Classical Lie
superalgebra; Cartan Lie superalgebra; generator

\vspace{0.05cm} \noindent \textbf{Mathematics Subject Classification
2000}: 17B05, 17B20, 17B70}
\end{quotation}
 \setcounter{section}{-1}
\section{Introduction}
Our principal aim is to determine the minimal number of generators
for a finite-dimensional simple Lie superalgebra over an
algebraically closed field of characteristic 0. The present work is
dependent on the classification theorem  due to Kac \cite{Kac},
which states that a simple Lie superalgebra (excluding simple Lie
algebras) is isomorphic to either a   classical Lie superalgebra or
a  Cartan Lie superalgebra (see also \cite{MS}).
In 2009, Bois  \cite{MBJ} proved that a simple Lie
algebra in arbitrary characteristic $p\neq 2,3$ is generated by  2
elements. In 1976, Ionescu \cite{GIT} proved that a simple Lie
algebra $L$ over the field of  complex numbers is generated by $1.5$
elements, that is, given any nonzero $x,$ there exists $y\in L$ such
that the pair $(x,y)$ generates  $L.$ In 1951, Kuranashi \cite{GKM} proved
that a semi-simple Lie algebra in characteristic 0 is  generated
by 2 elements.

As mentioned above, all the simple Lie superalgebras split into two
series: Classical Lie superalgebras and Cartan Lie superalgebras.
The Lie algebra (even part) of a classical Lie superalgebra is
reductive and meanwhile there exists a similarity in the structure side
between the Cartan Lie superalgebras in characteristic 0 and the
simple graded Lie algebra of Cartan type in characteristic $p$. Thus,
motivated by Bois's paper \cite{MBJ} and in view of the observation
above,  we began   this work in 2009. In the process we benefit in addition
much from the literatures above, especially from \cite{MBJ}, which
contains a
 considerable amount of information in   characteristic 0 and characteristic $p$.
We also use certain information about classical Lie superalgebras from \cite{Z}.

 Throughout we work over an algebraically
closed field $\mathbb{F}$ of characteristic 0 and all the vector spaces and algebras are finite
dimensional. The \textsc{main result} is that
 \textit{any simple Lie superalgebra is generated by 2 elements}.

\section{Classical Lie superalgebras}
\subsection{Basics}
 A  classical Lie superalgebra by definition is  a simple Lie superalgebra
 for which
  the representation of its Lie algebra (its even part) on the odd part is
completely reducible \cite{Kac,MS}. Throughout  this section, we
always write $L=L_{\bar{0}}\oplus L_{\bar{1}}$ for a classical
  Lie superalgebra. Our aim is to determine the minimal
 number of generators for a classical   Lie superalgebra $L$.  The strategy is as follows.
 First, using the results in Lie algebras \cite{MBJ,GIT}, we show that the Lie algebra $L_{\bar{0}}$ is generated
by 2 elements. Then, from the structure of  semi-simple Lie algebras
and their simple modules, we prove that each classical
  Lie superalgebra
 is generated by 2 elements.

  A classical   Lie  superalgebra is determined by its Lie algebra in a sense.
\begin{proposition}\cite[p.101, Theorem 1]{MS}\label{pro-even-part-reductive}
A simple Lie superalgebra  is classical if and only if its Lie
algebra
 is reductive.
\end{proposition}

   The following facts including Table 1.1 may be found  in
\cite{Kac,MS}.  The odd part $L_{\bar{1}}$ as
$L_{\bar{0}}$-module is completely reducible and $L_{\bar{1}}$   decomposes  into at most two irreducible components. By
Proposition \ref{pro-even-part-reductive}, $L_{\bar{0}}=C(L_{\bar{0}})\oplus
[L_{\bar{0}},L_{\bar{0}}].$ If the center $C(L_{\bar{0}})$ is
nonzero, then $\dim C(L_{\bar{0}})=1 $ and
$L_{\bar{1}}=L_{\bar{1}}^{1}\oplus L_{\bar{1}}^{2} $ is a direct sum
of two irreducible $L_{\bar{0}}$-submodules.
 For further information the reader is refereed to   \cite{MS, Kac}.\\

\begin{tabular}{|l|l|l|}
\multicolumn{3}{c}{\mbox{Classical Lie superalgebras (Table 1.1)}}

\\[5pt]\hline
\multicolumn{1}{|c|}{$L$}&\multicolumn{1}{|c|}{$L_{\bar{0}}$}&\multicolumn{1}{|c|}{$L_{\bar{1}}$ as $L_{\bar{0}}$-module}\\
\hline
 ${\rm{A}}(m,n),\; m,n\geq 0, n\neq m$&
${\rm{A}}_{m}\oplus {\rm{A}}_{n}\oplus
\mathbb{F}$&$~~\mathfrak{sl}_{m+1}\otimes\mathfrak{sl}_{n+1}\otimes \mathbb{F}\oplus (\mbox{its dual})$\\
\hline
 ${\rm{A}}(n,n),\;n> 0$&
${\rm{A}}_{n}\oplus {\rm{A}}_{n}$&$~~\mathfrak{sl}_{n+1}\otimes\mathfrak{sl}_{n+1}\oplus (\mbox{its dual})$\\
\hline ${\rm{B}}(m,n),\;m\geq 0,n>0$&${\rm{B}}_{m}\oplus
{\rm{C}}_{n}$&$~~\mathfrak{so}_{2m+1}\otimes
 \mathfrak{sp}_{2n}$\\
 \hline ${\rm{D}}(m,n),\;m\geq 2,n>0$&${\rm{D}}_{m}\oplus
{\rm{C}}_{n}$&~~$\mathfrak{so}_{2m}\otimes
 \mathfrak{sp}_{2n}$\\
  \hline
 ${\rm{C}}(n),\;n\geq 2$&${\rm{C}}_{n-1}\oplus \mathbb{F}$&$~~\mathfrak{csp}_{2n-2}\oplus (\mbox{its dual})$\\
 \hline
  ${\rm{P}}(n),\;n\geq 2$&${\rm{A}}_{n}$&
  $~~\Lambda^{2}\mathfrak{sl}^{*}_{n+1}\oplus{\rm{S}}^{2}\mathfrak{sl}_{n+1}$\\ \hline
 ${\rm{Q}}(n),\;n\geq 2$&
 ${\rm{A}}_{n}$& $~~{\rm{ad}}\mathfrak{sl}_{n+1}$\\ \hline
${\rm{D}}(2,1;\alpha),\;\alpha\in \mathbb{F}\setminus \{-1,0\}$&
 ${\rm{A}}_{1}\oplus
 {\rm{A}}_{1}\oplus {\rm{A}}_{1}$&
 $~~\mathfrak{sl}_{2}\otimes\mathfrak{sl}_{2}\otimes\mathfrak{sl}_{2}$\\ \hline
${\rm{G}}(3)$& $\mathfrak{G}_{2}\oplus
 {\rm{A}}_{1}$&~~$\mathfrak{G}_{2}\otimes
 \mathfrak{sl}_{2}$\\ \hline
 ${\rm{F}}(4)$&
${\rm{B}}_{3}\oplus  {\rm{A}}_{1}$&~~$\mathfrak{spin}_{7}\otimes
\mathfrak{sl}_{2}$\\ \hline
\end{tabular}

\subsection{Even parts}
Let $\mathfrak{g}$ be a semi-simple Lie algebra. Consider the  root
decomposition relative to a Cartan subalgebra $\mathfrak{h}$:
$\mathfrak{g}=\mathfrak{h}\oplus\bigoplus_{\alpha\in
\Phi}\mathfrak{g}^{\alpha}.$ For $x\in \mathfrak{g}$ we write  $
x=x_{\mathfrak{h}}+\sum_{\alpha\in\Phi} x^{\alpha}$ for  the
corresponding root space decomposition. It is well-known that \cite{H}
\begin{eqnarray}
  &&{\rm{dim}}\mathfrak{g}^{\alpha}=1 \quad  \mbox{for all}\;  \alpha\in \Phi,\label{eq1739f}\\
  &&\mathfrak{h}=\sum_{\alpha\in
\Phi}[\mathfrak{g}^{\alpha},\mathfrak{g}^{-\alpha}],\label{eq17394}\\
  && [\mathfrak{g}^{\alpha},\mathfrak{g}^{\beta}]=\mathfrak{g}^{\alpha+\beta}
  \quad\mbox{whenever}\; \alpha,\beta, \alpha+\beta\in \Phi.\label{eq0943}
\end{eqnarray}

Let $V$ be a vector space and $\mathfrak{F}:=\{f_{1},\ldots,f_{n}\}$
a finite set of non-zero linear functions on $V$. Write
$$\Omega_{\mathfrak{F}}:=\{v\in V\mid \Pi_{1\leq i\neq j\leq
n}(f_{i}-f_{j})(v)\neq 0\}.$$

\begin{lemma}\label{lemma-zarisk} Suppose $\mathfrak{F}$ is a finite  set of non-zero functions  in $V^{*}$. Then $\Omega_{\mathfrak{F}}\neq\emptyset.$ If $\mathfrak{G}\subset \mathfrak{F}$ then
     $\Omega_{\mathfrak{F}}\subset\Omega_{\mathfrak{G}}.$
\end{lemma}
\begin{proof} The first statement is from  \cite[Lemma
2.2.1]{MBJ} and the second is straightforward.
\end{proof}
This lemma  will be usually used  in the special situation when $V$
is a Cartan subalgebra of a simple Lie superalgebra.

    An element $x$ in a semi-simple Lie algebra $\mathfrak{g}$
is called \textit{balanced} if it has no zero components with
respect to the standard decomposition of
 simple Lie algebras. If $\mathfrak{h}$ is a Cartan subalgebra of $\mathfrak{g}$,
 $x\in\mathfrak{g}$ is called $\mathfrak{h}$-\textit{balanced}
 provided that  $x^{\alpha}\neq0$ for all $\alpha\in \Phi.$

\begin{lemma}\cite{MBJ}\label{lem1142lt}
An element of  a semi-simple Lie algebra  $\mathfrak{g}$ is balanced
if and only if it is $\mathfrak{h}$-balanced for some Cartan
subalgebra  $\mathfrak{h}$.
\end{lemma}
\begin{proof} One direction is obvious. Suppose $x\in\mathfrak{g}$
is balanced and let $\mathfrak{h}'$ is a Cartan subalgebra of
$\mathfrak{g}$. From the proof of \cite[Theorem 2.2.3]{MBJ}, there
exists $\varphi\in \mathfrak{g}$ such that $\varphi(x)$ is
$\mathfrak{h}'$-balanced. Letting
$\mathfrak{h}=\varphi^{-1}(\mathfrak{h}')$, one sees that
$\mathfrak{h}$ is a Cartan subalgebra and $x$ is
$\mathfrak{h}$-balanced.
\end{proof}

For an algebra $\mathfrak{A}$ and $x, y\in \mathfrak{A}$, we write
$\langle x,y\rangle$ for  the subalgebra  generated by $x$ and $y$.
We should notice that for a Lie superalgebra $\langle x,y\rangle$ is
not necessarily a $\mathbb{Z}_2$-graded subalgebra (hence   not
necessarily a sub-Lie superalgebra). The following technical lemma will be frequently used.

\begin{lemma}\label{lemmeigenvector} Let $\mathfrak{A}$ be an
algebra. For $a\in \mathfrak{A}$ write $L_a$ for the
left-multiplication operator given by $a$. Suppose
$x=x_{1}+x_{2}+\cdots+x_{n}$ is a sum of eigenvectors of $L_a$
associated with mutually distinct eigenvalues. Then all $x_{i}$'s
lie in the subalgebra generated $\langle a,x\rangle$.
\end{lemma}
\begin{proof} Let $\lambda_i$ be the eigenvalues of $L_a$ corresponding to
$x_i$. Suppose for a moment that all the $\lambda_i$'s are nonzero.
Then
$$
(L_a)^{k}(x)=\lambda_{1}^{k}x_{1}+\lambda_{2}^{k}x_{2}+\cdots+\lambda_{n}^{k}x_{n}\quad
\mbox{for}\; k\geq 1.
$$
Our conclusion in this case follows from the fact that the Vandermonde
 determinate given by
 $\lambda_1,\lambda_2,\ldots,\lambda_n$ is nonzero and thereby the
 general situation is clear.
\end{proof}

We write down a lemma from \cite[Theorem B and Corollory 2.2.5]{MBJ}
and the references therein,
  which is also a  consequence of Lemmas \ref{lemma-zarisk}, \ref{lem1142lt} and \ref{lemmeigenvector}.
\begin{lemma}\label{lem-semi-simple14}
Let $\mathfrak{g}$ be a semi-simple Lie algebra. If
$x\in\mathfrak{g}$ is balanced   then for a suitable Cartan
subalgebra $\mathfrak{h}$ and the corresponding root system $\Phi$
 we have  $\mathfrak{g}=\langle x, h\rangle$ for  all $h\in \Omega_{\Phi}$.
\end{lemma}

Denote by $\Pi:=\{\alpha_{1},\ldots,\alpha_{n}\}$ the system of
simple roots of a semi-simple Lie algebra $\mathfrak{g}$ relative to
a Cartan subalgebra $\mathfrak{h}$. As above, $x\in\mathfrak{g}$ is
refereed to as $\Pi$-balanced if $x$ is a sum of all the simple-root
vectors, that is, $x=\sum_{\alpha\in\Pi}x^{\alpha},$ where
$x^{\alpha}$ is a root vector of $\alpha$. Recall that
$\Omega_{\Pi}\neq\emptyset $ by Lemma \ref{lemma-zarisk}.
\begin{corollary}\label{lem-simple roots}
 A semi-simple Lie algebra $\mathfrak{g}$ is generated by a $\Pi$-balanced element and an
  element in $\Omega_{\Pi}$.
\end{corollary}
\begin{proof}
This is a consequence of  Lemma \ref{lemmeigenvector} and the
facts   (\ref{eq1739f}), (\ref{eq17394}) and (\ref{eq0943}).
\end{proof}

\begin{proposition}\label{pro-even-part-generators} The Lie algebra of a classical
  Lie superalgebra is generated by 2 elements.
\end{proposition}
\begin{proof}
Let $L=L_{\bar{0}}\oplus L_{\bar{1}}$ be a classical   Lie
superalgebra. By Proposition \ref{pro-even-part-reductive},
$L_{\bar{0}}$ is reductive, that is,  $[L_{\bar{0}},L_{\bar{0}}]$ is
semi-simple and
\begin{equation}\label{eq1853}
L_{\bar{0}}=C(L_{\bar{0}})\oplus[L_{\bar{0}},L_{\bar{0}}].
\end{equation}
If $C(L_{\bar{0}})=0$, the conclusion follows immediately from Lemma
\ref{lem-semi-simple14}. If $C(L_{\bar{0}})$ is nonzero, then
$C(L_{\bar{0}})=\mathbb{F}z$ is 1-dimensional. Choose a balanced
element $x\in [L_{\bar{0}},L_{\bar{0}}]$. By Lemma
\ref{lem-semi-simple14}, there exists $h\in
[L_{\bar{0}},L_{\bar{0}}]$ such that
$[L_{\bar{0}},L_{\bar{0}}]=\langle x,h\rangle.$ Claim that
$$L_{\bar{0}}=\langle x,h+z\rangle.$$
Indeed, considering
 the projection  of $L_{\bar{0}}$ onto $[L_{\bar{0}},L_{\bar{0}}] $ with respect to the decomposition
 (\ref{eq1853}), denoted by $\pi$, which is a homomorphism of Lie
algebras, we  have
$$\pi(\langle x,h+z\rangle)=\langle \pi(x),\pi(h+z)\rangle=\langle x,h\rangle=[L_{\bar{0}},L_{\bar{0}}].$$
Hence only two possibilities  might occur: $\langle
x,h+z\rangle=L_{\bar{0}}$ or ${\rm{dim}}\langle
x,h+z\rangle={\rm{dim}}[L_{\bar{0}},L_{\bar{0}}].$ The first case is
the desired. Let us show that the second does not occur. Assume the
contrary. Then
 $\pi$ restricting to $\langle x,h+z\rangle$ is an isomorphism and
 thereby
 $\langle x,h+z\rangle$ is semi-simple. Thus
$$\langle x,h+z\rangle=[\langle x,h+z\rangle,\langle x,h+z\rangle]=
[\langle x,h\rangle,\langle x,h\rangle]=\langle x,h\rangle.$$ Hence
  $h\in \langle
x,h+z\rangle$. It follows that
$$z\in \langle
x,h+z\rangle=\langle x,h\rangle=[L_{\bar{0}},L_{\bar{0}}],$$
contradicting  (\ref{eq1853}).
\end{proof}
\begin{remark}\label{remarkgln}
By Corollary \ref{lem-simple roots}, $\mathfrak{sl}(n)$ is generated
by a $\Pi$-balanced element $x$ and an
  element $y$ in $\Omega_{\Pi}$. As in the proof of Proposition
 \ref{pro-even-part-generators},
 one may prove that
$\mathfrak{gl}(n)$ is generated by $h$ and $x+z$,   where $z$ is a
nonzero central element in $\mathfrak{gl}(n)$.
\end{remark}

\subsection{Classical   Lie superalgebras}
Suppose $L$ is a  classical   Lie superalgebra with the standard Cartan subalgebra $H $. The corresponding weight (root) space decompositions are
\begin{eqnarray}
&& L_{\bar{0}}=H\oplus
\bigoplus_{\alpha\in \Delta_{\bar{0}}}L_{\bar{0}}^{\alpha},\qquad
L_{\bar{1}}=\bigoplus_{\beta\in
\Delta_{\bar{1}}}L_{\bar{1}}^{\beta};\nonumber\\
&& L=H\oplus \bigoplus_{\alpha\in
\Delta_{\bar{0}}}L_{\bar{0}}^{\alpha}\oplus\bigoplus_{\beta\in
\Delta_{\bar{1}}}L_{\bar{1}}^{\beta}.\label{eq-root-space-decompose}
\end{eqnarray}
Every $x\in L$ has a unique  decomposition with respect to (\ref{eq-root-space-decompose}):
\begin{equation}\label{eq-elememt-decompose}
x=x_{H}+\sum_{\alpha\in
\Delta_{\bar{0}}}x_{\bar{0}}^{\alpha}+\sum_{\beta\in
\Delta_{\bar{1}}}x_{\bar{1}}^{\beta},
\end{equation}
  where $x_{H}\in H,$ $x_{\bar{0}}^{\alpha}\in L_{\bar{0}}^{\alpha}$,
   $x_{\bar{1}}^{\beta}\in L_{\bar{1}}^{\beta}$.
Write
   $$
   \Delta:=\Delta_{\bar{0}}\cup \Delta_{\bar{1}}\quad\mbox{and}\quad L^{\gamma}:=L_{\bar{0}}^{\gamma}\oplus L_{\bar{1}}^{\gamma}\quad\mbox{for}\quad\gamma\in \Delta.
   $$
   Note that the standard  Cartan subalgebra of a classical Lie superalgebra is diagonal:
\begin{equation}\label{eq-classical-root-vector}
 \mathrm{ad}h(x)=\gamma(h)x \quad \mbox{for all}\;\; h\in H,\; x\in L^{\gamma}, \;\gamma\in \Delta.
 \end{equation}
 For $x\in L$, write
\begin{equation}
\mathbf{supp}(x):=\{\gamma\in\Delta \mid x_{\gamma}\neq
0\}.\label{eq-component}
\end{equation}
For $x=x_{\bar{0}}+x_{\bar{1}}\in L,$
$$\mathbf{supp}(x)=\mathbf{supp}(x_{\bar{0}})\cup\mathbf{supp}(x_{\bar{1}}).$$
\begin{lemma}\label{lem-weight-information}
~~~
\begin{itemize}
\item[$\mathrm{(1)}$] If $L\neq {\rm{Q}}(n)$ then $0\notin
\Delta_{\bar{1}}$ and $\Delta_{\bar{0}}\cap\Delta_{\bar{1}}=\emptyset.$

\item[$\mathrm{(2)}$]  If $L=
{\rm{Q}}(n)$ then $\Delta_{\bar{1}}=\{0\}\cup\Delta_{\bar{0}}.$

\item[$\mathrm{(3)}$]  If $L\neq\mathrm{A}(1,1)$, $\mathrm{Q}(n)$ or
$\mathrm{P}(3)$ then $\mathrm{dim}L^{\gamma}=1$ for every $\gamma\in
\Delta.$

\item[$\mathrm{(4)}$] Suppose $L={\rm{A}}(m,n) $, ${\rm{A}}(n,n), {\rm{C}}(n)$ or
${\rm{P}}(n),$ where $m\neq n$.

\begin{itemize}
\item[$\mathrm{(a)}$]
$L_{\bar{1}}=L_{\bar{1}}^{1}\oplus L_{\bar{1}}^{2}$ is a direct
sum of two irreducible $L_{\bar{0}}$-submodules.
 \item[$\mathrm{(b)}$] Let $\Delta_{\bar{1}}^{i}$ be  the  weight set of $L_{\bar{1}}^{i}$
relative to $H,$ $i=1,2$. Then there exist $\alpha_{\bar{1}}^{i}\in
 \Delta_{\bar{1}}^{i}$  such that
 $\alpha_{\bar{1}}^{1}\neq\alpha_{\bar{1}}^{2}.$
\end{itemize}
\end{itemize}
\end{lemma}
\begin{proof}
(1), (2) and (3) follow   from \cite[Proposition 1, p.137]{MS}.
  (4)(a) follows   from Table 1.1.  Let us consider (4)(b).
 For $L={\rm{A}}(m,n), {\rm{A}}(n,n)$ or
${\rm{C}}(n),$ it follows from the fact that   $L_{0}$-modules
$L_{-1}$ and $L_{1}$ are contragradient. For $L= {\rm{P}}(n),$ a
direct computation  shows that $-\varepsilon_{1}-\varepsilon_{2}\in
\Delta_{\bar{1}}^{1}$ and $2\varepsilon_{1}\in
\Delta_{\bar{1}}^{2}.$
 \end{proof}

\begin{theorem}\label{th-classical}
A  classical   Lie superalgebra is generated by 2 elements.
\end{theorem}

\begin{proof}
 Let $L=L_{\bar{0}}\oplus L_{\bar{1}}$
  be a classical Lie superalgebra.\\

  \noindent \textit{Case 1}. Suppose ${\rm{dim}}C(L_{\bar{0}})=1.$ In this case
  $L= {\rm{C}}(n) $ or ${\rm{A}}(m,n)$ with $m\neq n$ (see Table 1.1).  Then
$L_{\bar{1}}=L_{\bar{1}}^{1}\oplus L_{\bar{1}}^{2}$ is a direct sum
of two irreducible $L_{\bar{0}}$-submodules and
$[L_{\bar{0}},L_{\bar{0}}]$ is simple or a direct sum of  two simple
Lie algebras.
 Let $x_{\bar{0}}$ be a balanced element in
$[L_{\bar{0}},L_{\bar{0}}].$ From Lemma \ref{lem1142lt}, there
exists a Cartan subalgebra $\mathfrak{h}$ of
$[L_{\bar{0}},L_{\bar{0}}]$ such that
 ${\bf{supp}}(x_{\bar{0}})=\Delta_{\bar{0}},$ the latter is viewed as the root system relative to $\mathfrak{h}.$ By Lemma
 \ref{lem-semi-simple14}, we have
  $[L_{\bar{0}},L_{\bar{0}}]= \langle
x_{\bar{0}},h\rangle$ for all $h\in \Omega_{\Delta_{\bar{0}}}$.
Furthermore, from the proof of Proposition
\ref{pro-even-part-generators} it follows that $L_{\bar{0}}=\langle
x_{\bar{0}},h+z\rangle $ for $0\neq z\in C(L_{\bar{0}}).$
 By Lemma \ref{lem-weight-information}(1) and (4),  there exist
 $\alpha_{\bar{1}}^{1}\in \Delta_{\bar{1}}^{1}$ and $
 \alpha_{\bar{1}}^{2}\in \Delta_{\bar{1}}^{2}$ such that $\alpha_{\bar{1}}^{1}\neq
 \alpha_{\bar{1}}^{2}$ and $\alpha_{\bar{1}}^{1}, \alpha_{\bar{1}}^{2}\notin\Delta_{\bar{0}}.$
 Set
$x:=x_{\bar{0}}+x_{\bar{1}}^{\alpha_{\bar{1}}^{1}}+x_{\bar{1}}^{\alpha_{\bar{1}}^{2}}+z
$ for some weight vectors
 $  x_{\bar{1}}^{\alpha_{\bar{1}}^{i}}\in L_{\bar{1}}^{\alpha_{\bar{1}}^{i}},\; i=1,2.$
 Then
 $$x=(x_{\mathfrak{h}}+z)+\sum_{\alpha\in
\Delta_{\bar{0}}}x_{\bar{0}}^{\alpha}+x_{\bar{1}}^{\alpha_{\bar{1}}^{1}}+x_{\bar{1}}^{\alpha_{\bar{1}}^{2}}.$$
Write $\Phi:=\Delta_{\bar{0}}\cup
\{\alpha_{\bar{1}}^{1}\}\cup\{\alpha_{\bar{1}}^{2}\} $ and choose an
element $h'\in \Omega_{\Phi}.$ Assert $\langle x,h'\rangle=L.$ To
show that, write $ L':=\langle x,h'\rangle.$ Lemma
\ref{lemmeigenvector} implies all components $
x_{\bar{0}}^{\alpha},$ $x_{\bar{1}}^{\alpha_{\bar{1}}^{1}}$,
$x_{\bar{1}}^{\alpha_{\bar{1}}^{2}}$ and $x_{H}+z$ belong to $L'$.
Since $x_{\bar{0}}^{\alpha}\in L'$ for all $\alpha\in
\Delta_{\bar{0}},$
 from (\ref{eq17394}) we have  $x_{H}\in L' $ and then
 $z\in L'.$ As
$h'\in \Omega_{\Phi}\subset \Omega_{\Delta_{\bar{0}}},$ we obtain
 $\langle x_{\bar{0}},h'+z\rangle=L_{\bar{0}}\subset L'.$ Since $x_{\bar{1}}^{\alpha_{\bar{1}}^{i}}\in
L'$ and $L_{\bar{1}}^{\alpha_{i}}$ is an irreducible
$L_{\bar{0}}$-module,
 we have $L_{\bar{1}}^{i}\subset L',$ where $i=1,2.$ Therefore, $L=L'.$
\\

\noindent \textit{Case 2}. Suppose $C(L_{\bar{0}})=0.$ Then
$L_{\bar{0}}$ is a semi-simple Lie algebra and $L_{\bar{1}}$
decomposes into at most two irreducible components (see  Table 1.1).
\\

 \noindent \textit{Subcase 2.1}. Suppose $L_{\bar{1}}$ is an
 irreducible $L_{\bar{0}}$-module.
Note that in this subcase,  $L$ is of type $\mathrm{B}(m,n)$,
$\mathrm{D}(m,n)$, ${\rm{D}}(2,1;\alpha),$ $\mathrm{Q}(n),$
$\mathrm{G}(3)$ or $\mathrm{F}(4).$  We choose a weight vector $
x_{\bar{1}}^{\alpha_{\bar{1}}}\in L_{\bar{1}}^{\alpha_{\bar{1}}} $ ($\alpha_{\bar{1}}\neq 0$) and any balanced element $x_{\bar{0}}$ in
$L_{\bar{0}}.$ By Lemma \ref{lem1142lt}, we may assume that
 ${\bf{supp}}(x_{\bar{0}})=\Delta_{\bar{0}}.$

If $L\neq {\rm{Q}}(n),$    according to Lemma
\ref{lem-weight-information}(1), $
\alpha_{\bar{1}}\notin\Delta_{\bar{0}}.$    Let
$x=x_{\bar{0}}+x_{\bar{1}}^{\alpha_{\bar{1}}}.$ Then
$$x=x_{H}+\sum_{\alpha\in{\Delta_{\bar{0}}}}x_{\bar{0}}^{\alpha}+x_{\bar{1}}^{\alpha_{\bar{1}}}$$ is the
 root-vector decomposition. Let
$\Phi=\Delta_{\bar{0}}\cup\{\alpha_{\bar{1}}\}$. By Lemmas
\ref{lemma-zarisk} and \ref{lemmeigenvector},
  all components $x_{H}$, $x_{\bar{0}}^{\alpha}$
and  $x_{\bar{1}}^{\alpha_{\bar{1}}}$ belong to $\langle x,
h\rangle$ for $h\in \Omega_{\Phi}\subset H$. By (\ref{eq1739f}) and
(\ref{eq17394}), this yields $L_{\bar{0}}=\langle x_{\bar{0}},
h\rangle\subset\langle x, h\rangle.$ Since
$x_{\bar{1}}^{\alpha_{\bar{1}}}\in  \langle x, h\rangle$ and
$L_{\bar{1}}$ is irreducible as $L_{\bar{0}}$-module, we have
$L=\langle x, h\rangle.$

 Suppose $L={\rm{Q}}(n).$ Denote by
$\Pi:=\{\delta_{1},\delta_{2},\ldots,\delta_{n}\}$ the set of simple
roots of $L_{\bar{0}}$ relative to the Cartan subalgebra $H.$
According to Lemma \ref{lem-weight-information}(2), without loss of
generality we may assume that $
\alpha_{\bar{1}}:=\delta_{1}+\delta_{2}.$ Let
$x=x_{\bar{0}}+x_{\bar{1}}^{\alpha_{\bar{1}}}.$ Then
$$x=x_{H}+\sum_{\alpha\in{\Delta_{\bar{0}}}\setminus\{\alpha_{\bar{1}}\}}
x_{\bar{0}}^{\alpha}+(x_{\bar{0}}^{\alpha_{\bar{1}}}+x_{\bar{1}}^{\alpha_{\bar{1}}}).$$
  By Lemma \ref{lemmeigenvector},  all
components $x_{H}$, $x_{\bar{0}}^{\alpha}$ ($\alpha\in
\Delta_{\bar{0}}\setminus \{\alpha_{\bar{1}}\}$), and
$x_{\bar{0}}^{\alpha_{\bar{1}}}+x_{\bar{1}}^{\alpha_{\bar{1}}}$ belong to
$\langle x,h\rangle$, where $h\in \Omega_{\Delta_{\bar{0}}}\subset H.$
From (\ref{eq0943}) and (\ref{eq1739f})  we conclude that
$x_{\bar{0}}^{\alpha_{\bar{1}}}\in\mathbb{F}
[x_{\bar{0}}^{\delta_{1}},x_{\bar{0}}^{\delta_{2}}]\subset\langle
x,h\rangle $ and
then $x_{\bar{1}}^{\alpha_{\bar{1}}}\in\langle x,h\rangle.$ As above, the irreducibility of $L_{{\bar{1}}}$ yields  $L=\langle x,h\rangle.$\\

 \noindent
\textit{Subcase 2.2}. Suppose $L_{\bar{1}}=L_{\bar{1}}^{1}\oplus
L_{\bar{1}}^{2}$ is a direct sum of two irreducible
$L_{\bar{0}}$-submodules. In this case, $L={\rm{A}}(n,n)$ or $
{\rm{P}}(n).$
  Choose any balanced element $x_{\bar{0}}\in
L_{\bar{0}}$ and weight vectors
$x_{\bar{1}}^{\alpha_{\bar{1}}^{i}}\in
L_{\bar{1}}^{\alpha_{\bar{1}}^{i}},$  where ${\alpha_{\bar{1}}^{1}}$
and ${\alpha_{\bar{1}}^{2}}$ are different nonzero weights and
$\alpha_{\bar{1}}^{i}\notin \Delta_{\bar{0}}$ (Lemma
\ref{lem-weight-information}(1) and (4)).
 Lemma \ref{lem1142lt} allows us to assume that
 ${\bf{supp}}(x_{\bar{0}})=\Delta_{\bar{0}}.$ Let $x:= x_{\bar{0}}+x_{\bar{1}}^{\alpha_{\bar{1}}^{1}}+x_{\bar{1}}^{\alpha_{\bar{1}}^{2}}$
 and $\Phi:=\Delta_{\bar{0}}\cup
\{\alpha_{\bar{1}}^{1}\}\cup \{\alpha_{\bar{1}}^{2}\}.$ As before, we are able to deduce that
$L_{\bar{0}}\subset \langle x,h\rangle$ and
$x_{\bar{1}}^{\alpha_{\bar{1}}^{1}},x_{\bar{1}}^{\alpha_{\bar{1}}^{2}}\in \langle x,h\rangle$ for
$h\in\Omega_{\Phi}\subset
\Omega_{\Delta_{\bar{0}}}\subset H.$ Thanks to the irreducibility of $L_{\bar{1}}^{1}$ and $L_{\bar{1}}^{2}$,
 we have $L=\langle x,h\rangle $. The proof is complete.
\end{proof}

\begin{remark}
In view of the proof of Theorem \ref{th-classical}, starting from
any balanced element in the semi-simple part of the Lie algebra of a
classical Lie superalgebra $L$ we are able to find two elements
generating $L.$
\end{remark}
By Theorem \ref{th-classical}, as in the proof of Proposition
\ref{pro-even-part-generators}, one is able to  prove the following
\begin{corollary}
The general linear Lie superalgebra $\mathfrak{gl}(m,n)$ is
generated by 2 elements.
\end{corollary}

As a subsidiary result, let us show  that a classical Lie superalgebra, except for
$\mathrm{A}(1,1)$, $\mathrm{Q}(n)$ or $\mathrm{P}(3)$, is generated
by 2 \textit{homogeneous} elements. By Lemma \ref{lem-weight-information}(3), for such a classical Lie superalgebra, all the odd-weight subspaces are 1-dimensional. Here we give a more general description in Remark \ref{homogeneous generators}. As before, an element $x\in L$ is called $\Delta_{\bar{1}}$-balanced if $x$ is
a sum of all the odd-weight vectors, namely,
$x=\sum_{\gamma\in\Delta_{\bar{1}}}x_{\bar{1}}^{\gamma},$ where
$x_{\bar{1}}^{\gamma}$ is a weight vector of $\gamma$.
\begin{remark}\label{homogeneous generators}
A finite dimensional simple Lie superalgeba (not necessarily classical) for which all the odd-weight is $1$-dimensional is generated by  2 homogeneous elements.
\end{remark}
\begin{proof} Let $L$ be such a Lie superalgebra.
Choose  a $\Delta_{\bar{1}}$-balanced element
$x=\sum_{\gamma\in\Delta_{\bar{1}}}x_{\bar{1}}^{\gamma}$ and any
$h\in \Omega_{\Delta_{\bar{1}}}\subset H.$ By Lemmas
\ref{lemma-zarisk} and \ref{lemmeigenvector},
  all components $x_{\bar{1}}^{\gamma}$ belong to $\langle x,
h\rangle$ for $h\in \Omega_{\Delta_{\bar{1}}}\subset H$. Since
$\mathrm{dim}L^{\gamma}=1,$ we conclude that
$L^{\gamma}\subset\langle x, h\rangle$ for all $\gamma\in
\Delta_{\bar{1}}.$ By \cite[Proposition
1.2.7(1), p.20]{Kac}, $L_{\overline{0}}=[L_{\overline{0}},L_{\overline{0}}]$ and then $\langle x,
h\rangle=L$.
\end{proof}

Finally  we give an example to explain how to find the pairs of generators  in Theorem
\ref{th-classical} and Remark \ref{homogeneous generators}.

\begin{example} Let $\mathrm{A}={\rm{A}}(1;0)$. Find the generators of $\mathrm{A}$ as in Theorem
\ref{th-classical} and Remark \ref{homogeneous generators}.
\end{example}
  Recall that $ {\rm{A}} =\{x\in  \mathfrak{gl}(2;1)\mid
{\rm{str}}(x)=0\}$.
 Its Lie algebra is a direct sum of the $1$-dimensional center and the
semi-simple part:
$${\rm{A}} _{\bar{0}}=\mathbb{F}(e_{11}+e_{22}+2e_{33})\oplus
[{\rm{A}} _{\bar{0}},{\rm{A}} _{\bar{0}}],$$ where $[{\rm{A}}
_{\bar{0}},{\rm{A}}
_{\bar{0}}]=\mathrm{span}_{\mathbb{F}}\{e_{11}-e_{22}, e_{12},
e_{21}\}.$ The odd part is a direct sum of two irreducible
$\mathrm{A} _{\bar{0}}$-submodules:
\begin{eqnarray*}
{\rm{A}} _{\bar{1}}={\rm{A}} _{\bar{1}}^{1}\oplus {\rm{A}}
_{\bar{1}}^{2}=\mathrm{span}_{\mathbb{F}}\{e_{13},e_{23}\}\oplus
\mathrm{span}_{\mathbb{F}}\{ e_{31},e_{32}\}.\end{eqnarray*} The
standard Cartan subalgebra is
$H=\mathrm{span}_\mathbb{F}\{e_{11}-e_{22},e_{11}+e_{22}+2e_{33}\}.$

Table 1.2 gives all the roots and the corresponding root vectors.
\\

~~\begin{tabular}{|l|l|l|l|l|l|l|}
\multicolumn{7}{c}{ Table 1.2} \\[1pt]
\hline
 { roots}&
$\varepsilon_{1}-\varepsilon_{2}$&$\varepsilon_{2}-\varepsilon_{1}$&$\varepsilon_{1}-2\varepsilon_{3}$&$\varepsilon_{2}-2\varepsilon_{3}$&$-\varepsilon_{1}+2\varepsilon_{3}$&$-\varepsilon_{2}+2\varepsilon_{3}$\\
\hline { vectors}&$\hfill e_{12}\hfill$&$\hfill e_{21}\hfill$&$\hfill e_{13}\hfill$&$\hfill e_{23}\hfill$&$\hfill e_{31}\hfill$&$\hfill e_{32}\hfill$\\
\hline
\end{tabular}\\
\begin{itemize}
\item \textit{Theorem \ref{th-classical}-Version}.
Put $x:=(e_{12}+e_{21})+e_{13}+e_{31}+(e_{11}+e_{22}+2e_{33})$ and
$h:=3e_{11}+e_{22}+4e_{33}.$ From Table 1.2, the weight values
corresponding to $e_{12},e_{21},e_{13},e_{31}$ are $2, -2, -5, 5,$
respectively. As in the proof of Theorem \ref{th-classical}, we have
$$e_{12},e_{21},e_{13},e_{31},e_{11}+e_{22}+2e_{33}\in \langle
x,h\rangle.$$ Furthermore,
\begin{eqnarray*}&&\langle e_{12}+e_{21}, h+(e_{11}+e_{22}+2e_{33})\rangle=
{\rm{A}}_{\bar{0}}\subset \langle x,h\rangle.\end{eqnarray*} Since
${\rm{A}}_{\bar{1}}^{i}$ is an irreducible
${\rm{A}}_{\bar{0}}$-module,
  ${\rm{A}}_{\bar{1}}^{i}\subset\langle x,h\rangle$,  $i=1,2.$ Hence $\mathrm{A}=\langle x,h\rangle$.

\item \textit{Remark \ref{homogeneous generators}-Version}.
Consider the $\Delta_{\overline{1}}$-balanced element
$x:=e_{13}+e_{31}+e_{23}+e_{32}$ and write $h:=e_{11}+e_{33}.$ By Table
1.2, the weight values corresponding to
$e_{13},e_{31},e_{23},e_{32}$ are $-1, 1, -2, 2,$ respectively. As in
the proof of Remark \ref{homogeneous generators}, we have
$e_{13},e_{31},e_{23},e_{32}\in \langle x,h\rangle.$ Since
$\mathrm{dimA}_{\bar{1}}^{\lambda}=1$ for $\lambda\in
\Delta_{\bar{1}}$ and $[
\mathrm{A}_{\bar{1}},\mathrm{A}_{\bar{1}}]=\mathrm{A}_{\overline{0}},$ we
obtain $\mathrm{A}=\langle x,h\rangle.$
\end{itemize}

\section{Cartan Lie superalgebras}
All the Cartan Lie superalgebras are listed below \cite{Kac, MS}:
\begin{itemize}
\item[]
$W(n)$ ($n\geq 3$),\;  $S(n)$ ($n\geq 4$),\;
$\widetilde{S}(2m)$ ($m\geq 2$),\; $H(n)$
  ($n\geq 5$).
  \end{itemize}
Let $\Lambda(n)$ be the
Grassmann superalgebra with $n$  generators $\xi_{1},\ldots,\xi_{n}$.
For a $k$-\textit{shuffle} $u:=(i_{1},i_{2},\ldots,i_{k})$, that is, a
strictly increasing sequence between $1$ and $n$,  we write $|u|:=k$
and $x^{u}:=\xi_{i_{1}}\xi_{i_{2}} \cdots \xi_{i_{k}}.$ Letting
${\rm{deg}}\xi_{i}=1,\; i=1,\ldots,n,$ we obtain the so-called standard
$\mathbb{Z}$-grading of $\Lambda(n).$ Let us briefly
describe the
  Cartan  Lie superalgebras.

  \begin{itemize}\item
 $W(n)={\rm{der}}\Lambda(n)$ is $\mathbb{Z}$-graded,
$W(n)=\oplus_{k=-1}^{n-1}W(n)_{k},$
$$W(n)_{k}={\rm{span}}_{\mathbb{F}}\{
x^{u}\partial/\partial\xi_{i}\mid |u|=k+1,\; 1\leq i\leq n\}.$$
\end{itemize}
\begin{itemize}\item
$S(n)=\oplus_{k=-1}^{n-2}S(n)_{k}$ is a $\mathbb{Z}$-graded
subalgebra of $W(n)$,
$$S(n)_{k}={\rm{span}}_{\mathbb{F}}\{\mathrm{D}_{ij}(x^{u})\mid |u|=k+2,\,\ 1\leq i,j\leq n\}.$$
 Hereafter, ${{\mathrm{D}_{ij}}}(f):=\partial(f)/
\partial\xi_{i}\partial/\partial\xi_{j}+\partial(f)/\partial\xi_{j}\partial/\partial\xi_{i}$
for $f\in \Lambda(n).$
\end{itemize}
\begin{itemize}\item $\widetilde{S}(2m)$ ($m\geq 2$) is a subalgebra of $W(2m)$ and
as  a $\mathbb{Z}$-graded subspace,
$$\widetilde{S}(2m)=\oplus_{k=-1}^{2m-2}\widetilde{S}(2m)_{k},$$
where
\begin{eqnarray*}&&\widetilde{S}(2m)_{-1}={\rm{span}}_{\mathbb{F}}\{(1+\xi_{1}\cdots\xi_{2m})\partial/\partial\xi_{j}\mid
1\leq j\leq 2m\},\\&&\widetilde{S}(2m)_{k}=S(2m)_{k},\; 0\leq k\leq
2m-2.\end{eqnarray*} Notice that $\widetilde{S}(2m)$ is not a
$\mathbb{Z}$-graded subalgebra of $W(2m)$.
 \end{itemize}
 \begin{itemize}
 \item
 $H(n)=\oplus_{k=-1}^{n-3}H(n)_{k}$ is a
      $\mathbb{Z}$-graded
subalgebra of $W(n)$, where
$$H(n)_{k}={\rm{span}}_{\mathbb{F}}\{{\rm{D_{H}}}(x^{u})\mid|u|=k+2\}.$$
To explain the linear mapping ${\rm{D_{H}}}: \Lambda(n)\longrightarrow W(n)$, write $n=2m$ $(m\geq 3)$
or $2m+1$ $(m\geq 2).$  By definition,
 ${\rm{D_{H}}}(x^{u}):=(-1)^{|u|}\sum_{i=1}^{n}\partial(x^{u})/\partial\xi_{i}\partial/\partial\xi_{i'} $
  for any  shuffle   $u,$ where
 $'$ is the
involution of the index set $\{1,\ldots,n\}$ satisfying that $i'=i+m$ for
      $i\leq m$.
 \end{itemize}

For simplicity we usually write $W, S, \widetilde{S}, H$ for $W(n),
S(n), \widetilde{S}(2m), H(n),$ respectively. Throughout this section  $L$
denotes one of  Cartan Lie superalgebras. Consider its decomposition of
subspaces mentioned above:
\begin{equation}\label{eqcartangraded}
L=L_{-1}\oplus \cdots \oplus L_{s}.
\end{equation}
For $W, S, \widetilde{S}$ and $H$, the height $s$ is $n-1,$ $n-2,$ $2m-2$
or $n-3,$ respectively. Note that $S $ and $H$ are
$\mathbb{Z}$-graded subalgebras of $W$ with respect to
(\ref{eqcartangraded}), but $\widetilde{S}$ is not. The null $L_{0}$
is isomorphic to $\mathfrak{gl}(n),
\mathfrak{sl}(n),\mathfrak{sl}(2m),\mathfrak{so}(n)$ for $L=W, S,
\widetilde{S}, H,$ respectively.

\begin{lemma}\label{lem-Cartan-component}
~~~
\begin{itemize}
\item[$\mathrm{(1)}$]  $L_{-1}$ and
$L_{s}$ are irreducible as $L_{0}$-modules.
\item[$\mathrm{(2)}$] $L_{1}$
is an irreducible $L_{0}$-module for $L=S,\widetilde{S}$ or $H,$ except for $H(6).$ For $L=H(6),$
$L_1$ is a direct sum of
two irreducible $L_0$-submodules.

\item[$\mathrm{(3)}$] $L$ is generated by  the local part
$L_{-1}\oplus L_{0}\oplus L_{1}.$

\item[$\mathrm{(4)}$] $L$ is generated by
$L_{-1}$ and $L_{s}$ for $L=W$, $S $ or $H$.
\end{itemize}
\end{lemma}
\begin{proof} All the statements are standards (see \cite{Kac,MS} for example), except for
that $\widetilde{S}_{-1}$ is  irreducible as
$\widetilde{S}_{0}$-module.
 Indeed, a direct verification shows that $\widetilde{S}_{-1}$ is an $\widetilde{S}_{0}$-module and the
irreducibility follows from the canonical isomorphism of $S_{0}$-modules
$\varphi: S_{-1} \longrightarrow  \widetilde{S}_{-1}$ assigning $
\partial/\partial\xi_{i}$ to
  $(1+\xi_{1}\cdots\xi_{2m})\partial/\partial\xi_{i} $ for $1\leq i\leq 2m.$
\end{proof}

The following is a list of
bases of the standard Cartan subalgebras $\mathfrak{h}_{L_0} $ of
$L_{0}.$\\

~~~~~~~~~~~\begin{tabular}{|l|l|}
\multicolumn{2}{c}{Table 2.1} \\[1pt]\hline
\multicolumn{1}{|c|}{$L$}&\multicolumn{1}{|c|}{ basis of $\mathfrak{h}_{L_0}$}\\
\hline
 ~~$W(n)$&
~~~$\xi_{i}\partial/\partial\xi_{i},$ $1\leq
i\leq n$\\
\hline ~~$S(n)$&~~~$
\xi_{1}\partial/\partial\xi_{1}-\xi_{j}\partial/\partial\xi_{j},$
$2\leq j\leq n$\\
\hline ~~$\widetilde{S}(2m)$&~~~$
\xi_{1}\partial/\partial\xi_{1}-\xi_{j}\partial/\partial\xi_{j},$
$2\leq j\leq 2m$\\
\hline
~~$H(2m)$&~~~$\xi_{i}\partial/\partial\xi_{i}-\xi_{m+i}\partial/\partial\xi_{m+i},$
$1\leq i\leq m$\\
\hline
~~$H(2m+1)$&~~~$\xi_{i+1}\partial/\partial\xi_{i+1}-\xi_{m+i}\partial/\partial\xi_{m+i},$
$1\leq i\leq m$\\
\hline
\end{tabular}\\

\noindent The weight space decomposition of the component $L_k$ relative to
   $\mathfrak{h}_{L_0}$ is:
$$L_{k}=\delta_{k,0}\mathfrak{h}_{L_0}\oplus_{\alpha\in
\Delta_{k}}L_{k}^{\alpha},\,\ \mbox{where}\,\ -1\leq k\leq s.$$
By Lemma \ref{lem-Cartan-component}(2), $H(6)_{1}$ is a direct sum
of two irreducible $H(6)_{0}$-modules
$$H(6)_{1}=H(6)_{1}^{1}\oplus H(6)_{1}^{2}.$$ Let $\Delta_{1}^{i}$
be the weight set of $H(6)_{1}^{i},$ $i=1,2.$
 Write $\Pi$ for  the set of simple roots of  $L_0$ relative to the
Cartan subalgebra $\mathfrak{h}_{L_0}$. We have

\begin{lemma}\label{lem-cartan-weight-information}
~
\begin{itemize}
\item[$\mathrm{(1)}$]
If $L=W$ or $S$ then
$\Pi\cap\Delta_{-1}=\Pi\cap\Delta_{s}=\Delta_{-1}\cap
\Delta_{s}=\emptyset.$
\item[$\mathrm{(2)}$]
If  $L=\widetilde{S}$   then
$\Pi\cap\Delta_{-1}=\Pi\cap\Delta_{1}=\Delta_{-1}\cap
\Delta_{1}=\emptyset.$
\item[$\mathrm{(3)}$]
If  $L=H(2m)$  then $\Pi\cap\Delta_{-1}=\Pi\cap\Delta_{1}=\emptyset$
and $\Delta_{-1}\neq \Delta_{1}$.
\item[$\mathrm{(4)}$]
 If $L=H(2m+1) $ then  $0\in \Delta_{-1},$  $\Pi\neq \Delta_{1}$ and $\Delta_{-1}\neq\Delta_{1}.$
\item[$\mathrm{(5)}$] There exist
nonzero weights $\alpha_{1}^{i}\in \Delta_{1}^{i} $  such
that $\alpha_{1}^{1}\neq \alpha_{1}^{2}.$
\end{itemize}
\end{lemma}
\begin{proof}
 We first compute the  weight sets of the desired components and the
 system of
 simple roots of $L_0.$
For  $W(n)$,
\begin{eqnarray*}
&&\Delta_{-1}=\{-\varepsilon_{j}\mid 1\leq j\leq
n\},\qquad~~~~~~~~\Delta_{0}=\{\varepsilon_{i}-\varepsilon_{j}\mid 1\leq i\neq
j\leq n\},\\
&&\Pi=\{\varepsilon_{i}-\varepsilon_{i+1}\mid 1\leq i\leq
n-1\},\qquad\Delta_{s}=\bigg\{\sum_{k=1}^{n}\varepsilon_{k}-\varepsilon_{j}\mid
1\leq j\leq n\bigg\}.
\end{eqnarray*}
For $S(n)$ and
$\widetilde{S}(n),$
\begin{eqnarray*}
&&\Delta_{-1}=\{-\varepsilon_{j}\mid 1\leq j\leq
n\},\qquad ~~~~~~~~\Delta_{0}=\{\varepsilon_{i}-\varepsilon_{j}\mid 1\leq i\neq j\leq
n\},\\
&&\Pi=\{\varepsilon_{i}-\varepsilon_{i+1}\mid 1\leq i\leq
n-1\},\qquad \Delta_{1}=\big\{\varepsilon_{k}+\varepsilon_{l}-\varepsilon_{j}\mid
1\leq k, l,j\leq n\big\},\\
&&\Delta_{s}=\bigg\{\sum_{i=1}^{n}\varepsilon_{i}-\varepsilon_{j}-\varepsilon_{k}\mid
1\leq j,k\leq n\bigg\}.\end{eqnarray*}
For $H(2m)$,
\begin{eqnarray}
&&\Delta_{-1}=\{\ \pm\varepsilon_{j}\mid 1\leq j\leq
m\},\qquad\Delta_{0}=\{\pm(\varepsilon_{i}+\varepsilon_{j}),\pm(\varepsilon_{i}-\varepsilon_{j})
\mid 1\leq i<j\leq m\},\nonumber\\
&&\Pi=\{\varepsilon_{i}-\varepsilon_{i+1},
\varepsilon_{m-1}+\varepsilon_{m}\mid 1\leq i< m\},\nonumber\\
&&\Delta_{1}=\{\pm(\varepsilon_{i}+\varepsilon_{j})\pm\varepsilon_{k},
\pm(\varepsilon_{i}-\varepsilon_{j})\pm\varepsilon_{k}\mid  1\leq i<
j< k\leq m\}\nonumber\\&&~~~~~~~~\cup\{\pm\varepsilon_{l}\mid 1\leq
l\leq m\}.\label{eqheven23}
\end{eqnarray}
For $H(2m+1)$, write $\varepsilon_{i}'=\varepsilon_{i+1}$ for $1\leq
i\leq m.$ We have
\begin{eqnarray*}
&&\Delta_{-1}=\{0\}\cup \{\pm\varepsilon_{i}'\mid 1\leq i\leq m\},\\
&&\Delta_{0}=\{\pm\varepsilon_{k}',\pm(\varepsilon_{i}'+\varepsilon_{j}'),\pm(\varepsilon_{i}'-\varepsilon_{j}')\mid
1\leq k\leq m, 1\leq i< j\leq m\},\\
&&\Pi=\{\varepsilon_{i}'-\varepsilon_{i+1}',\varepsilon_{m}'\mid
1\leq i< m\},\\
&&\Delta_{1}=\{0\}\cup\{\pm\varepsilon_{l}',\pm(\varepsilon_{i}'+\varepsilon_{j}'),\pm(\varepsilon_{i}'-\varepsilon_{j}')\mid
1\leq l\leq m, 1\leq i< j\leq m\}\\
&&~~~~~~~~\cup\{\pm(\varepsilon_{i}'+\varepsilon_{j}')\pm\varepsilon_{k},
\pm(\varepsilon_{i}'-\varepsilon_{j}')\pm\varepsilon_{k}')\mid 1\leq
i< j< k\leq m\}.\end{eqnarray*}

  All the statements follow directly, except  (5)  for $L=H(6).$
  In this special case, from (\ref{eqheven23})  one sees that $0\notin \Delta_{1} $ and $|\Delta_{1}|>1$.
  Consequently,  (5) holds.
\end{proof}

Recall that an element $x\in\mathfrak{g}$ is refereed to as
$\Pi$-balanced if $x$ is a sum of all the simple-root vectors.

\begin{theorem}\label{th-main-cartan}
 A Cartan Lie superalgebra is generated by 2
elements.
\end{theorem}
\begin{proof}
  Recall the null $L_{0}$ is
isomorphic to $\mathfrak{gl}(n), \mathfrak{sl}(n),
\mathfrak{sl}(2m)$ or $ \mathfrak{so}(n)$. From Remark
\ref{remarkgln} and Corollary \ref{lem-simple roots}, for a
$\Pi$-balanced element $x_{0}\in L_{0}$ and
$h_{0}\in\Omega_{\Pi}\subset \mathfrak{h}_{L_0}$ we have
$L_{0}=\langle x_{0}+\delta_{L,W}z,h_{0}\rangle,$ where $z $ is a
central element in $\mathfrak{gl}(n).$

For simplicity,    write $t:=s$ for $L=W$ or $S$ and $t:=1$ for
$L=\widetilde{S}$ or $H$.  Suppose $L\neq
H(6)\;\mbox{and}\;H(2m+1).$ According to
 Lemma \ref{lem-cartan-weight-information},  we are able to choose nonzero weights
 $\alpha_{-1}\in \Delta_{-1}$ and $\alpha_{t}\in \Delta_{t} $ such that
 $\alpha_{-1}\neq\alpha_{t} $, $\alpha_{-1}\notin{\Pi},$ and
  $\alpha_{t}\notin{\Pi}.$ Put
$x:=x_{-1}+x_{0}+\delta_{L,W}z+x_{t} $ for some weight vectors
$x_{-1}\in L_{-1}^{\alpha_{-1}}$ and $x_{t}\in L_{t}^{\alpha_{t}}.$
Now set $\Phi:=\Pi\cup\{\alpha_{-1}\}\cup\{\alpha_{t} \}\subset
\mathfrak{h}_{L_0}^{*}$ and choose an element $h_{0}\in
\Omega_{\Phi}.$ Assert $\langle x,h_{0}\rangle=L.$  Lemma
\ref{lemmeigenvector} implies all components $ x_{-1}$ $x_{0}$,
$\delta_{L,W}z$ and $x_{t}$ belong to $\langle x,h_{0}\rangle.$ As
$h_{0}\in \Omega_{\Phi}\subset \Omega_{\Pi},$ we obtain
 $L_{0}=\langle x_{0}+\delta_{L,W}z,h_{0}\rangle \subset \langle x,h_{0}\rangle.$ By Lemma
 \ref{lem-Cartan-component}(1) and (2), since  $L_{-1}$ and $L_{t}$ are irreducible
$L_0$-modules,
 we have $L_{-1}+L_{t}\subset \langle x,h_{0}\rangle.$  From Lemma \ref{lem-Cartan-component}(3) and (4) it follows that $L=\langle
 x,h_{0}\rangle.$

 If $L=H(6),$ by Lemma
\ref{lem-cartan-weight-information}(3), we are able to choose $\alpha_{-1}\in
\Delta_{-1},$ $\alpha_{1}^{1}\in \Delta_{1}^{1}$ and
$\alpha_{1}^{2}\in \Delta_{1}^{2}$ such that
$\alpha_{-1},\alpha_{1}^{1},\alpha_{1}^{2}$ are pairwise distinct
and $\alpha_{-1}\notin \Pi,$ $\alpha_{1}^{1}\notin \Pi$ and
$\alpha_{1}^{2}\notin \Pi$. Put
$x:=x_{-1}+x_{0}+x_{1}^{1}+x_{1}^{2}$ for some weight vectors  $
x_{-1}\in L_{-1}^{\alpha_{-1}} $ and $ x_{1}^{i}\in
L_{1}^{\alpha_{1}^{i}},$  $i=1,2.$   Write
$\Phi:=\Pi\cup\{\alpha_{-1}\}\cup\{\alpha_{1}^{1}\}\cup\{\alpha_{1}^{2}
\}.$ For $h_{0}\in \Omega_{\Phi}\subset \Omega_{\Pi}$, as in the
above, one may show that $L=\langle x,h_{0}\rangle$.

If $L=H(2m+1),$ by Lemma \ref{lem-cartan-weight-information}(4),
choose $\alpha_{-1}\in \Delta_{-1},$ $\alpha_{1}\in \Delta_{1}$
 such that $\alpha_{-1}=0,$ $\alpha_{1}\notin \Pi.$
   Set $x:=x_{-1}+x_{0}+x_{1}$ for some weight
vectors $ x_{-1}\in L_{-1}^{\alpha_{-1}}$ and $x_{1}\in
L_{t}^{\alpha_{1}}.$
 Now put
$\Phi:=\Pi\cup\{\alpha_{-1}\}\cup\{\alpha_{1} \}\subset
\mathfrak{h}_{L_0}^{*}.$ Let $h_{0}\in \Omega_{\Phi}\subset
\Omega_{\Pi} $ and  claim that $L=\langle x,h_{0}\rangle.$ By Lemma
\ref{lemmeigenvector}, $x_{0},$ $x_{-1}$ and $x_{1}\in \langle
x,h_{0}\rangle.$  Consequently, $L_{0}\subset L $. The
irreducibility of $L_{-1}$ and $L_{1}$ ensures $L_{-1}+L_{1}\subset \langle x,h_{0}\rangle.$ By Lemma
\ref{lem-Cartan-component}(3), the claim holds. The proof is complete.
\end{proof}

 Theorems \ref{th-classical} and \ref{th-main-cartan} combine to the main result of this paper:
\begin{theorem}
Any simple Lie superalgebra is generated by 2 elements.
\end{theorem}

\end{document}